\documentclass{llncs}
\usepackage{epsfig,amsfonts,latexsym,amssymb,amsmath}

\def\OnlineSearch{\mbox{{\tt LeapFrog}}}
\def\SwarmSpeed{\mbox{{\tt SwarmSpeed}}}
\def\exit{\mbox{{\tt EXIT}}}
\def\walk{\mbox{{\tt WALK}}}
\def\search{\mbox{{\tt SEARCH}}}
\def\bp{Beachcombers' Problem}
\def\obp{Online Beachcombers' Problem}

\def\ca{\mbox{{\tt Comb}}}

\usepackage{amsfonts,latexsym,graphicx,epsfig,amssymb,color}
\usepackage{rotating}

\usepackage{makeidx} 
\usepackage{algorithm}
\usepackage{algorithmic}

\usepackage{etex,etoolbox}
\usepackage{comment}

\newcommand{\costis}[1]{{\color{blue}\textsl{\small[#1]}\marginpar{\small\textsc{\textbf{Costis-Comment}}}}}

\makeindex

\newcommand{\ignore}[1]{}
\DeclareMathOperator{\argmax}{arg max}

\def \reals {\mathbb{R}}

\begin{document}
\sloppy
\title {The Beachcombers' Problem:\\ Walking and Searching 
with Mobile Robots}

\author{Jurek Czyzowicz~\inst{1} \and Leszek Gasieniec~\inst{2} \and Konstantinos Georgiou~\inst{3} \and Evangelos Kranakis~\inst{4} \and Fraser MacQuarrie~\inst{4}}

\institute{Universit\'{e} du Qu\'{e}bec en Outaouais, Department d'Informatique, Gatineau, Qu\'{e}bec, Canada. 
\and
University of Liverpool, Department of Computer Science, Liverpool, UK.
\and
University of Waterloo, Dept. of Combinatorics \& Optimization, Waterloo, Ontario, Canada.
\and
Carleton University, School of Computer Science, Ottawa, Ontario, Canada.}

\maketitle

\begin{abstract}
We introduce and study a new problem concerning the exploration of
a geometric domain by mobile robots. Consider a line segment $[0,I]$
and a set of $n$ mobile robots  $r_1,r_2,\ldots , r_n$ 
placed at one of its endpoints. Each robot has a {\em searching speed} $s_i$ 
and a {\em walking speed} $w_i$, where $s_i <w_i$. We assume that each robot is aware of the number of robots of the collection and their corresponding speeds. 
At each time moment a robot  $r_i$ either walks along a portion of the segment not exceeding its walking speed $w_i$ or searches a portion of the segment with the speed not exceeding $s_i$. A search of  segment $[0,I]$ is completed at the time when each of its points have been searched by at least one of the $n$ robots.
We want to develop {\em mobility schedules} (algorithms)
for the robots which complete the search of the segment as fast as possible. More exactly we want to maximize the {\em speed} of the mobility schedule (equal to the ratio of the segment length versus the time of the completion of the schedule).


We analyze first the offline scenario when the robots know the length of the segment that is to be searched. We give 
an algorithm producing a mobility schedule for arbitrary walking and searching speeds and prove its optimality. Then we propose an online algorithm, when the robots do not know in advance the actual length of the segment to be searched. The speed $S$ of such algorithm is defined as
$$
S = \inf_{I_L} S(I_L)
$$
where $S(I_L)$ denotes the speed of searching of segment $I_L=[0,L]$. We prove that the proposed online algorithm is 2-competitive. The competitive ratio is shown to be better in the case when the robots' walking speeds are all the same.


\vspace{0.2cm}
\noindent
{\bf Key words and phrases.}
Algorithm,
Mobile Robots,
On-line, 
Schedule,
Searching,
Segment,
Speed,
Walking.
\end{abstract}

\section{Introduction}

A domain being a segment of known or unknown length has to be explored collectively
by $n$ mobile robots initially placed in a segment endpoint. At every time moment a  robot may perform either of the two different activities of {\em walking} and {\em searching}. 
While walking, each robot may
{\em traverse} the domain with a speed not exceeding its maximal {\em walking speed}. During searching, the robot performs a more 
{\em elaborate} task on the domain.  The bounds on the walking and searching speeds may be different for different robots, but we always assume that each robot can walk with greater maximal speed than it can search.
Our goal is to design the movement of all robots so that each point of the domain is being searched by at least one robot and the time when the process is completed is minimized (i.e. the speed of the process is maximized). 

In many situations {\em two-speed} searching is a convenient way to approach exploration of various domains. For example {\em foraging} or {\em harvesting} a field may take longer than walking across. Intruder searching activity takes more time than uninvolved territory traversal. In computer science {\em web pages indexing, forensic search, code inspection, packet sniffing} require more involved inspection process. Similar problems arise in many other domains. We call our question the {\em Beachcombers' Problem} to show up the analogy to the situation when each mobile searcher looking for an object of value in the one-dimensional domain proceeds slower when searching rather than while simply performing an unconcerned traversal of the domain.

In our problem, the searchers collaborate in order to terminate the searching process as quickly as possible. Our algorithms generate {\em mobility schedules} i.e. sequences of moves of the agents, which assure that every point of the environment is inspected by at least one agent while this agent was performing the searching activity.

\subsection{Preliminaries}

Let $I_L$ denote the interval $[0,L]$ for any positive integer $L$. Consider $n$ robots $r_1,r_2,\ldots, r_n$,
each robot $r_i$ having searching speed $s_i$ and walking speed $w_i$, such that $s_i < w_i$. A {\em searching schedule} ${\cal A}$ of $I_L$ is defined by an increasing sequence of time moments $t_0=0, t_1, \dots , t_z$, such that in each time interval $[t_j, t_{j+1}]$ every robot $r_i$ either walks along some subsegment of  $I_L$ not exceeding its walking speed $w_i$, or searches  some subsegment of  $I_L$ not exceeding its searching speed $s_i$. The searching schedule is {\em correct} if for each point $p \in I_L$ there is some $j \geq 1$ and some robot $r_i$, such that during the time interval $[t_j, t_{j+1}]$ robot $r_i$ searches the subsegment of $I_L$ containing point $p$. 

By the speed $S_{\cal A}(I_L)$ of schedule ${\cal A}$ searching interval $I_L$ we mean the value of $S_{\cal{A}}(I_L) = L/t_z$. We call $t_z$ the {\em finishing time} of the searching schedule. The searching schedule is optimal if there does not exist any other correct searching schedule having a speed larger than $S$. 

It is easy to see that the schedule speed maximization criterion is equivalent to its finishing time minimization when the segment length is given or to the searched segment length maximization when the time bound is set in advance. However the speed maximization criterion applies better to the online problem when the objective of the schedule is to perform searching of an unknown-length segment or a semi-line. Such schedule successively searches the intervals $I_L$ for the increasing values of  $L$. The speed of such schedule is defined as
$$
S_{\cal{A}} = \inf_{I_L} S_{\cal{A}}(I_L)
$$

Observe that any searching schedule may be converted to another one, which has the property that all subsegments which were being searched (during some time intervals $[t_j, t_{j+1}]$ by some robots) have pairwise disjoint interiors. Indeed, if some subsegment is being searched by two different robots (or twice by the same robot), the second searching may be replaced by the walk through it by the involved robot. Since the walking speed of any robot is always larger than its searching speed, the speed of such converted schedule is not smaller than the original one. Therefore, when looking for the optimal searching schedule, it is sufficient to restrict the consideration to schedules whose searched subsegments may only intersect at their endpoints. In the sequel, all searching schedules in our paper will have such property.

Notice as well, that, when looking for the most efficient schedule, we may restrict our consideration to schedules such that at any time moment a robot $r_i$ is either searching using its maximal searching speed $s_i$, or walking with maximal allowed speed $w_i$. Indeed, whenever $r_i$ searches (or walks) during a time interval $[t_j, t_{j+1}]$ using a non-maximal and not necessarily constant searching speed (resp. walking speed) we may replace it with a search (resp. walk) using maximal allowed speed. It is easy to see that the search time of any point, for such modified schedule, is never longer, so the speed of such schedule is not decreased.

We assume that all the robots start their exploration at the same time and that are able to cross over each other.

\begin{definition}[Beachcombers' Problem]\label{def: bp}
Consider an interval $I_L=[0,L]$ and $n$ robots $r_1,r_2,\ldots, r_n$, initially placed at its endpoint $0$, each robot $r_i$ having searching speed $s_i$ and walking speed $w_i$, such that $s_i < w_i$. The Beachcombers' Problem consist in finding an efficient correct searching schedule ${\cal A}$ of $I_L$. The speed $S_{\cal A}$ of the solution to the \bp~equals $S_{\cal A}=I_L/t_z$, where 
$t_z$ is the finishing time of {\cal A}. 
\end{definition}

We also study the online version of this problem:

\begin{definition}[Online Beachcombers' Problem]\label{def: bop} Consider $n$ robots $r_1,r_2,\ldots, r_n$, initially placed at the origine of a semi-line I, each robot $r_i$ having searching speed $s_i$ and walking speed $w_i$, such that $s_i < w_i$. The Online Beachcombers' Problem consist in finding a correct searching schedule {\cal A} of $I$. The cost $S_{\cal A}$ of the solution to the \obp, called the {\em speed of {\cal A}}, equals
$$
S_{\cal{A}} = \inf_{I_L} S_{\cal{A}}(I_L) = \inf_{I_L}\frac{I_L}{t_z(I_L)}
$$
where $I_L=[0,L]$ for any positive integer $L$ and
$t_z(I_L)$ denotes the time when the search of the segment $I_L=[0,L]$ is completed. 
\end{definition}

\subsection{Related Work}

The original text on graph searching started with the
work of Koopman~\cite{koopman1946search}. Many papers followed studying searching and exploration of graphs (e.g. \cite{DP,FT08}) or geometric environments, (e.g.\cite{Albers00,alpern2002theory,baeza1993searching,CILP13,DKP91}). The purpose of these studies was usually either to learn (map) an unknown environment (e.g.\cite{DP}) or to search it, looking for a target (motionless or mobile) (cf. \cite{FT08}). 

Many searching problems were studied from a game-theoretic viewpoint (see \cite{alpern2002theory}). \cite{alpern2002theory} presented an approach to searching and rendezvous, when two mobile players either collaborate in order to find each other, or they compete against each other - one willing to meet and the other one to avoid each other. Searching 1-dimensional environments (segments, lines, semi-lines), similarly to the present paper, despite the simplicity of the environment, often led to interesting results (cf. \cite{Bellman63,Beck64,demaine2006online}). 

The efficiency of the searching or exploration algorithm is usually measured by the time used by the mobile agent,  often proportional to the distance travelled.
Many searching and especially exploration algorithms are {\em online}, i.e. they concern a priori unknown environments, cf. \cite{Albers03,AS11}. Performance of such algorithms is expressed by {\em competitive ratio}, i.e. the proportion of the time spent by the online algorithm versus the time of the optimal offline algorithm, which assumes the knowledge of the environment (cf. \cite{berman1998line,FKKLT08}). Most exploration algorithms (e.g. \cite{CILP13,DKP91,DDKPU13} and several search algorithms (e.g. \cite{demaine2006online}) use the competitive ratio to measure their performance.

Most of the above research concerned single robots. Collections of mobile robots, collaborating in order to reduce the exploration time, were used, e.g., in \cite{CFMS10,DFKNS07,FGKP06,HKLT13}. Most recently \cite{DDKPU13} studied  tradeoffs between the number of robots and the time of exploration showing how a polynomial number of agents may search the graph optimally.

Some papers studying mobile robots assume distinct robot speeds. Varying mobile sensor speed was used in \cite{WIFBZP11} for the purpose of sensor energy efficiency. \cite{beauquier2010utilizing} was utilizing distinct agent speeds to design fast converging protocols, e.g. for gathering. \cite{CGKK11,KK12} considered distinct speeds for robots patrolling boundaries. However to the best of our knowledge, the present paper is the first one assuming two-speed robots for the problem of searching or exploration.

\subsection{Outline and Results of the Paper}

In Section~\ref{sec: exact polytime solution for ssc1} we begin by studying the properties of optimal schedules. We then propose  
{\em "comb"} algorithm, an optimal algorithm for \bp~which requires $O(n \log n)$ computational
steps, and prove its correctness.
Section~\ref{sec: online} is devoted to online searching, where the length of the segment to be searched is not known in advance. In this section we propose the  online searching algorithm \OnlineSearch, prove its correctness and analyze its efficiency. We prove that the \OnlineSearch~algorithm is 2-competitive. The competitive ratio is shown to be reduced to 1.29843 in the case when all robots' walking speeds are the same.
Section~\ref{sec: Conclusion and Open Questions}
concludes the paper and proposes  problems for further research.  Any proofs not given in the paper may be found in the Appendix.

\section{Searching a Known Segment}\label{sec: exact polytime solution for ssc1}

We proceed by first identifying in Section~\ref{sec: properties of opt} a number of structural properties exhibited by every optimal solution to the \bp.  This will allow us to conclude in Section~\ref{sec: the opt schedule} that \bp~can be solved efficiently.

\subsection{Properties of Optimal Schedules}\label{sec: properties of opt}


\begin{lemma}\label{lem: all properties}
Any optimal schedule for the \bp~may be converted to another optimal schedule, such that 
\begin{itemize}
\item[{(a)}] every robot searches a contiguous subinterval;
\item[{(b)}] at no time during the execution of this schedule is a robot idle, just before the finishing time all robots are searching, and they all finish searching exactly at the schedule finishing time;
\item[{(c)}] all robots are utilized, i.e. each of them searches a non-empty subinterval; 
\item[{(d)}] for any two robots $r_i,r_j$ with $w_i <w_j$, robot $r_i$ searches a subinterval closer to the starting point than the subinterval of robot $r_j$.
\end{itemize}
\end{lemma}

By applying these properties, we determine a useful recurrence for the subintervals robots search in an optimal schedule. 

\begin{lemma}\label{lem: rec relation of subintervals}
Let the robots $r_1, r_2, \ldots , r_n$ be ordered in non-decreasing walking speed, and suppose that $T_\textrm{opt}$ is the time of the optimal schedule. Then, 
\begin{enumerate}
\item The segment to be searched may be partitioned into successive subsegments of lengths $c_1, c_2, \ldots, c_n$ and the optimal schedule assigns to robot $r_i$ the $i$-th interval of length $c_i$, where
\item The length $c_i$ satisfies the following recursive formula, where we assume, without loss of generality, that $w_0=0$ and $w_1=1$.\footnote{\label{foo: why w0=0 and w1=1}We set $w_0=0$ and $w_1=1$ for notational convenience, so that~\eqref{equa: recursive relation of interval lengths} holds. Note that $w_0$ does not correspond to any robot, while $w_1$ is the walking speed of the robot that will search the first subinterval, and so will never enter walking mode, hence, $w_1$ does not affect our solution. }
\begin{equation}\label{equa: recursive relation of interval lengths}
c_0 = 0; \qquad c_k = \frac{s_k}{w_k} 
\left(
\left(
\frac{w_{k-1}}{s_{k-1}}-1
\right)c_{k-1}
+T_\textrm{opt}(w_k-w_{k-1})
\right),~~k \geq 1
\end{equation}

\end{enumerate}
\end{lemma}

\begin{proof}
From Lemma~\ref{lem: all properties}(a) we know that all robots must search contiguous intervals. Since by Lemma~\ref{lem: all properties}(c) we need to utilize all robots, it follows that the optimal schedule defines a partition of the unit domain into $n$ subintervals. Finally by Lemma~\ref{lem: all properties}(d), we know that if we order the robots in non-decreasing walking speed, then robot $r_i$ will search the $i$-th in a row interval, showing the first claim of the lemma. 

Now, from Lemma~\ref{lem: all properties}(b), we know that all robots finish at the same time, say $T$. Since all robots start processing the domain at the same time, robot $k$ will walk its initial subinterval of length $\sum_{i=1}^{k-1} c_i$ in time proportional to $1/w_k$, and in the remaining time it will search the interval of length $c_k$. Hence
$$ c_k = s_k \left(T - \frac{\sum_{i=1}^{k-1} c_i}{w_k}\right),$$
from which we easily derive the desired recursion.
\end{proof}

\subsection{The Optimal Schedule for the \bp}\label{sec: the opt schedule}

As a consequence of Lemma~\ref{lem: all properties} we have the following offline algorithm \ca~producing an optimal schedule. The algorithm is parameterized by the real values $c_i$ equal to the sizes of intervals to be searched by each robot $r_i$.

\begin{figure}[h]
\begin{center}
\fbox{
\begin{minipage}[t][1.9cm]
{12cm}
{\bf Algorithm \ca }; \\
1. \hspace*{0.12cm} Sort the robots in non-decreasing walking speeds; \\
2. \hspace*{0.12cm} {\bf for } $i \leftarrow 1$ {\bf to } $n$  {\bf do } \\
3. \hspace*{0.75cm} Robot $r_i$ first walks the interval of length $\sum_{j=1}^{i-1} c_j$, \\
\hspace*{1.2cm} and then searches interval of length $c_i$ \\
\end{minipage}
}
\end{center}
\end{figure}

We can now prove the following theorem:

\begin{theorem}\label{thm: 1-side searching}
The \bp~can be solved optimally in $O(n\log n)$ many steps.
\end{theorem}

\begin{proof}
By Lemma~\ref{lem: rec relation of subintervals} we need to order the robots by non-decreasing walking speed, which requires $O(n\log n)$ many steps). We then show how to compute all $c_i$ in linear number of steps, modulo the arithmetic operations that depend on the encoding sizes of $w_i, s_i$. 

Consider an imaginary unit time period.  Starting with the slowest, for each robot, we use~\eqref{equa: recursive relation of interval lengths} to compute (in constant time) the subinterval $y_i$ it would search if it were to remain active for the unit time period.  Consequently, we can compute in $n$ steps the total length $\sum_{i=1}^n y_i$ of the interval that the collection of robots can search within a unit time period. This schedule, scaled to a unit domain, will have finishing time $ T = 1 / \sum_{i=1}^n y_i.$  The length of the interval that robot $r_k$ will search is then $c_k = {y_k}/{\sum_{i=1}^n y_i}$. 
\end{proof}

\subsection{Closed Formulas for the Optimal Schedule of the \bp}\label{sec: closed formulas}

From the proof of Theorem~\ref{thm: 1-side searching} we can implicitly derive the time (and the speed) of an optimal solution to the \bp. In what follows, we assume that $w_i=0$, that the robots are ordered in non-decreasing walking speeds, 
 and that $w_1=1$ (see Lemma~\ref{lem: rec relation of subintervals} and Footnote~\ref{foo: why w0=0 and w1=1}). 

\begin{lemma}\label{lem: opt solution arb instances}
Consider a set of robots 
such that in the optimal schedule each robot finishes searching in time $T_\textrm{opt}$.  Robot $r_k$ will search a subinterval of length $c_k$, such that
\begin{equation}\label{equa: xk value}
\frac{c_k}{T_\textrm{opt}} =  s_k - \frac{s_k}{w_k} \sum_{r=1}^{k-1} s_r \prod_{j=r+1}^{k-1} \left( 1 - \frac{s_j}{w_j}\right) 
\end{equation}
\end{lemma}

\begin{proof}
To prove~\eqref{equa: xk value}, we need to show that the values specified for $c_k$ satisfy recurrence~\eqref{equa: recursive relation of interval lengths} which has a unique solution. Indeed, 
\begin{eqnarray*}
	\frac{w_k}{s_k}c_k 
&\stackrel{\eqref{equa: xk value}}{=}&
	T_\textrm{opt}\left(
	s_k - \frac{s_k}{w_k} \sum_{r=1}^{k-1} s_r \prod_{j=r+1}^{k-1} \left( 1 - \frac{s_j}{w_j}\right)
	\right) \\
&=&
T_\textrm{opt}( w_k - s_{k-1}) - T_\textrm{opt} \sum_{r=1}^{k-2} s_r \prod_{j=r+1}^{k-1} \left( 1 - \frac{s_j}{w_j}\right) \\
&=&
	T_\textrm{opt}( w_k - s_{k-1}) - T_\textrm{opt} \left( 1 - \frac{s_{k-1}}{w_{k-1}} \right) \sum_{r=1}^{k-2} s_r \prod_{j=r+1}^{k-2} \left( 1 - \frac{s_j}{w_j}\right) \\
&\stackrel{\eqref{equa: xk value}}{=}&
	T_\textrm{opt}( w_k - s_{k-1}) - \left( 1 - \frac{s_{k-1}}{w_{k-1}} \right) \frac{w_{k-1}}{s_{k-1}} ( T_\textrm{opt} s_{k-1} - c_{k-1}) \\
&=&
	T_\textrm{opt}( w_k - w_{k-1}) + \left( \frac{s_{k-1}}{w_{k-1}} - 1 \right)  c_{k-1}
\end{eqnarray*}
which is exactly~\eqref{equa: recursive relation of interval lengths}.
\end{proof}

\begin{definition}[Search Power] \label{def: search power}
Consider a set of $n$ robots $r_1, r_2, \dots , r_n$, with $s_i <w_i$, $i=1,\ldots,n$. We define the \textit{search power} of any subset  $A$ of robots using a real function $g:2^{[n]}\mapsto \reals^+$ as follows: For any subset $A$, first sort the items in non-decreasing walking speeds $w_i$, and let $w^A_1,\ldots,w^A_{|A|}$ be that ordering (the superscripts just indicate membership in $A$). We define the evaluation function (search power of set $A$) as 
$$ g(A) := \sum_{k=1}^{|A|} s_k^A \prod_{j=k+1}^{{|A|}} \left( 1 - \frac{s_j^A}{w_j^A} \right), $$
\end{definition}
Note that the search power of any subset of the robots is well defined, and that it is always positive (since $s_i < w_i$). By summing \eqref{equa: xk value} for $k=1,\ldots, n$ and using the identity $\sum_{k=1}^n c_k=1$, we obtain the following theorem:

\begin{theorem}\label{thm: opt solution to ssc1}
The speed $S_\textrm{opt}$ of the optimal schedule equals the search power of the collection of robots. In other words, if $N$ denotes the set of all robots, then 
\begin{equation}\label{equa: opt 1-side solution}
S_\textrm{opt}=  g(N).
\end{equation}
\end{theorem}

From Theorem~\ref{thm: opt solution to ssc1}, we can obtain the speed of the optimal schedule when all robots have the same walking speed. 

\begin{corollary}\label{cor: xs in w-uniform}
Let $s_1,\ldots,s_n$ be the speeds of robots where all walking speeds are 1. Then the speed $S_\textrm{opt}$ of the optimal schedule is given by the formula
\begin{equation}\label{equa: sol for w-uniform}
S_\textrm{opt} =  1 - \prod_{i =1}^n (1-s_i).
\end{equation}
which is exactly the simplified expression of the search power of such set of robots. 
\end{corollary}


\section{The Online Search Algorithm}\label{sec: online}

In this section we give an algorithm producing a searching schedule for a segment of size not known in advance to the robots. Each robot execute the same sequence of moves for each unit segment. Therefore, contrary to the offline case, in which all robots complete their searching duties at the same finishing time (at different positions), in the online algorithm the robots arrive all together at point 1 of the unit segment. Therefore the speed of searching of each integer segment is the same and we call it {\em swarm speed}.  However, the robots which cannot contribute to increase the overall swarm speed are not used in the schedule. Each used robot $r_i$ (called a {\em swarm robot}) searches a subsegment of the unit segment of size $c_i$ and walks along the remaining part of it. The subsegments $c_i$, whose lengths are chosen in order to synchronize the arrival of all robots at the same time at every integer point, are pairwise interior disjoint and they altogether cover the entire unit segment, i.e.$\sum_{i=1}^k c_i=1$.

Below we define the procedure {\tt SwarmSpeed} which
determines the speed of a swarm in linear time
and algorithm {\tt OnlineSearch} which defines the
swarm.
Algorithm {\tt OnlineSearch}, defines the schedule for
 a swarm of
$k$ robots $r_1,r_2,\ldots, r_k$
out of the original
$n$ robots 
such that $w_1 \geq w_2 \geq \cdots \geq w_k$.

\begin{figure}[h]
\begin{center}
\fbox{
\begin{minipage}[t][2.8cm]
{12cm}
{\bf real procedure \SwarmSpeed ()}; \\
1. \hspace*{0.12cm} \textbf{var} $S \leftarrow 0, \hspace{0.2cm} S_{num} \leftarrow 0, \hspace{0.2cm} S_{den} \leftarrow 1, \hspace{0.2cm} \delta$ {\bf: real}; \hspace{0.2cm} $i \leftarrow 1$ {\bf:  integer}; \\
2. \hspace*{0.12cm} {\bf while}  $i \leq n$ {\mbox and} $S < w_i$ {\bf do} \\
3. \hspace*{0.55cm}  $\delta \leftarrow 1/({\frac{1}{s_i}-\frac{1}{w_i}});$ \\
4. \hspace*{0.55cm} $S_{num} \leftarrow S_{num}   + \delta; \hspace{0.2cm} S_{den} \leftarrow S_{den} + \delta/w_i$;\hspace{0.2cm} $S=\frac{S_{num}}{S_{den}}$;\\
5. \hspace*{0.55cm} $i \leftarrow i+1;$\\
6. \hspace*{0.12cm} {\bf return}  $S$; \\
\end{minipage}
}
\end{center}
\end{figure}

Once the swarm speed has been computed, it is possible to compute the subsegments lengths $c_i$, that we call the {\em contribution} of robot $r_i$ - the fraction of the unit interval that $r_i$ is allotted to search.

\begin{figure}[h]
\begin{center}
\fbox{
\begin{minipage}[t][4.0cm]
{12cm}
{\bf Algorithm \OnlineSearch (robot $r_j$)}; \\
1. \hspace*{0.12cm} \textbf{var} $S \leftarrow \SwarmSpeed()$; \\
2. \hspace*{0.12cm} {\bf if } $w_j \leq S$ {\bf then} \\
3. \hspace*{0.75cm}$\exit;$ \{robot $r_j$ stays motionless\} \\
4. \hspace*{0.12cm} {\bf else }  \\
5. \hspace*{0.6cm} {\bf for } $i \leftarrow 1$ {\bf to } $j-1$  {\bf do } \\
6. \hspace*{0.95cm} \walk($(\frac{1}{s}-\frac{1}{w_i})/(\frac{1}{s_i}-\frac{1}{w_i})$);\\
7. \hspace*{0.55cm} {\bf while}  {\mbox not at line end} {\bf do} \\
8. \hspace*{0.95cm}  \search($(\frac{1}{s}-\frac{1}{w_j})/(\frac{1}{s_j}-\frac{1}{w_j})$);\\
9. \hspace*{0.95cm}  \walk($1-(\frac{1}{s}-\frac{1}{w_j})/(\frac{1}{s_j}-\frac{1}{w_j})$);\\
\end{minipage}
}
\end{center}
\end{figure}

\begin{theorem}
\label{online-speed}
Consider a partition of 
the unit interval into
$k$ consecutive
non-overlapping segments $C_1, C_2, \ldots ,C_k$, from left to right, of lengths 
$c_1, c_2, \ldots , c_k$, respectively. 
Assume that all the
robots start (at endpoint $0$) and finish (at endpoint $1$) 
simultaneously. Further assume that the $i$-th robot $r_i$
searches the segment $C_i$ with speed $s_i$ and walks
the rest of the interval $I \setminus C_i$ with speed $w_i$
such that $w_i > s_i$.
Then the speed of the swarm satisfies
\begin{equation}
\label{deta}
S  = 
\frac{ \sum_{i=1}^k \frac{1}{\delta_i}}{1+ \sum_{i=1}^k \frac{1}{w_i \delta_i}},
\end{equation}
where $\delta_i := \frac{1}{s_i} - \frac{1}{w_i}$, for $i=1,2,\ldots ,k$.
\end{theorem}
\begin{proof}
The partition of the interval $[0, 1]$ into segments as
prescribed in the statement of the lemma gives rise to the equation
\begin{equation}
\label{det0}
c_1 + c_2 + \cdots + c_k  =  1.
\end{equation} 
Let $s$ be the speed of the swarm of $n$ robots.
Since all the
robots must reach the other endpoint $1$ of the interval 
at the same time, 
we have the following identities.
\begin{equation}
\label{det1}
\frac{c_i}{s_i} + \frac{1-c_i}{w_i} = \frac{1}{S}, \mbox{ for $1 \leq i  \leq k$,}
\end{equation}
where $\frac{c_i}{s_i}$ is the time spent searching and $ \frac{1-c_i}{w_i}$
the time spent walking by robot $r_i$.
Using the notation 
\begin{equation}
\label{det2}
\delta_i := \frac{1}{s_i} - \frac{1}{w_i},
\end{equation} 
and substituting into Equation~\eqref{det1},
after simplifications
we get
\begin{equation}
\label{det3}
c_i = 
\frac{1}{S\delta_i} - \frac{1}{w_i \delta_i}, 
\mbox{ for $1 \leq i \leq k$.}
\end{equation}
Using Equation~\eqref{det0} we see that
$$
1=\sum_{i=1}^k c_i 
=\sum_{i=1}^k \frac{1}{S\delta_i} - \sum_{i=1}^k \frac{1}{w_i \delta_i}
=\frac{1}{S} \sum_{i=1}^k \frac{1}{\delta_i} - \sum_{i=1}^k \frac{1}{w_i \delta_i},
$$
which implies Identity~\eqref{deta}, as desired.
\end{proof}

\begin{lemma}
\label{detlm2}
Algorithm  {\tt OnlineSearch} is correct (i.e. every point of 
the semiline $[0,+\infty )$ is searched by a robot).
\end{lemma}
\begin{proof}
Let $C_j (i)$ denote the subsegment of $[i, i+1]$ of length $c_j$
which is searched by robot $r_j$. The lemma follows from the
observation that $\bigcup_{j=1}^k C_j (i) = [i, i+1]$, for all
$i \geq 0$ and all $j=1,2,\ldots ,k$.
\end{proof}

\section{Competitiveness of the Online Searching}
\label{sec: competitiveness}

In this section we discuss the competitiveness of the \OnlineSearch~algorithm. 
Since competitive ratio is naturally discussed more often for cost optimization (minimization) problems, we assume in this section that we compare the finishing time (rather than speed) of the online versus offline solution. We show first that in the general case the \OnlineSearch~Algorithm is 2-competitive.

\begin{theorem}
\label{th: competitive ratio}
Consider any set of robots $r_1, r_2, \ldots, r_n$, ordered by a non-decreasing walking speed. If the completion time of the optimal schedule produced by the \ca~algorithm equals $T_\textrm{opt}$ then the completion time $T_\textrm{online}$ of the searching schedule produced by the \OnlineSearch~algorithm is such that $T_\textrm{online} < 2T_\textrm{opt}$.
\end{theorem}
\begin{proof}
As \OnlineSearch~algorithm outputs schedules of the same speed for all integer-length segments it is sufficient to analyze its competitiveness for a unit segment. Assume, to the contrary, that the time $T_\textrm{online}$ of the schedule output by \OnlineSearch~is such that $T_\textrm{online} \geq 2T_c$. Note that, the swarm speed $S$ of the \OnlineSearch~is then at most $S \leq 1/(2T_\textrm{online})$.
Consider $C_1, C_2, \ldots, C_n$ - the subsegments searched by robots $r_1, r_2, \ldots, r_n$, respectively. Recall that each robot $r_i$ of the \ca~algorithm walks along segments $C_1, C_2, \ldots, C_{i-1}$ and searches $C_i$ arriving at its right endpoint at time $T_c{opt}$.
Let $i^*$ be the index such that the midpoint $1/2 \in C_{i^*}$ (or point $1/2$ is a common endpoint of $C_{i^*}$ and $C_{i^*+1}$). Observe, that in time $2T_\textrm{opt}$ each robot $r_i$, such that $i \geq i^*$ could reach the right endpoint of the unit segment, while searching its portion of length $2|C_i|$. Note that, as for each robot $r_i$, such that $i \geq i^*$, we have $w_i > 1/(2T_\textrm{online}) \geq S$, each such robot is used by \OnlineSearch~in lines 5-9. However, since $\sum_{i=i^*}^n 2|C_i| > 1$ all robots $r_i$, for $i \geq i^*$ search a segment longer than 1, arriving at its right endpoint within time $2T_\textrm{opt}$, or $T_\textrm{online} < 2T_\textrm{opt}$ for the unit segment. This contradicts the earlier assumption.
\end{proof}

Observe that, the competitive ratio of 2 may be approached as close as we want. Indeed, we have the following
\begin{proposition} For any $\epsilon >0$ there is a set of two robots for which the \OnlineSearch~algorithm produces a schedule of completion time $T_\textrm{online}$ such that $T_\textrm{online} = (2-\epsilon)T_\textrm{opt}$.
\end{proposition}
\begin{proof}
Let the speeds of the two robots be $s_1=1-\epsilon/2, w_1 = 1, s_2=1, w_2=(2-\epsilon)/\epsilon$. As the swarm speed $S$ computed in \SwarmSpeed~procedure equals 1, the line 2 of the \OnlineSearch~algorithm excludes $r_1$ from the swarm, so the search is performed uniquely by $r_1$ with $T_\textrm{online}=1$.  Using Theorem~\ref{thm: opt solution to ssc1} we get 
$$
S_\textrm{opt}= \sum_{k=1}^{2} s_k \prod_{j=k+1}^{{2}} \left( 1 - \frac{s_j}{w_j} \right)=
 \left( 1 - \frac{\epsilon}{2} \right)  \left( 1 - \frac{\epsilon}{2-\epsilon} \right) +1=
 2-\epsilon
$$
Hence $T_\textrm{opt}=1/(2-\epsilon)$ and $T_\textrm{online} =1= (2-\epsilon)T_\textrm{opt}$
\end{proof}

The following theorem concerns the competitiveness of the \OnlineSearch~algorithm in the special case when all robot walking speeds are the same. 

\begin{theorem}\label{thm: summary of online w-uniform}
Let be given the collection of robots $r_1, r_2, \ldots, r_n$ with the same walking speed $w=w_1=\ldots = w_n$. The \OnlineSearch~algorithm has the competitive ratio $\alpha_n$ which is increasing in $n$. In particular, $\alpha_2 = 1.115$, $\alpha_3 \approx 1.17605$, $\alpha_4 \approx 1.20386$ and $\lim_{n \rightarrow \infty} \alpha_n \approx 1.29843$. 
\end{theorem}

Our strategy towards proving Theorem~\ref{thm: summary of online w-uniform} is to show that the competitive ratio of \OnlineSearch~-among all problem instances when walking speeds are the same - is maximized when all robots' searching speeds are also the same. Because of lack of space, the section~\ref{sec: online with w-uniform} related to the proof of Theorem \ref{thm: summary of online w-uniform} is entirely deferred to the Appendix.

\section{Conclusion and Open Problems}
\label{sec: Conclusion and Open Questions}

In this paper, we proposed and analyzed 
offline and online algorithms for addressing the 
beachcombers' problem. The offline algorithm, when the
size of the segment to search is known in advance is shown to 
produce the optimal schedule. The online searching algorithm is shown to be 2-competitive in general case and 1.29843-competitive when the agents' walking speeds are known to be the same. We conjecture that there is no online algorithm with the competitive ratio of $(2-\epsilon)$ for any $\epsilon >0$.

Other open questions concern different domain topologies, robots starting to search from different initial positions or the case of faulty robots. 

\bibliography{biblio}

\begin{thebibliography}{10}

\bibitem{koopman1946search}
Koopman, B.O.:
\newblock Search and screening.
\newblock Operations Evaluation Group, Office of the Chief of Naval Operations,
  Navy Department (1946)

\bibitem{DP}
Deng, X., Papadimitriou, C.H.:
\newblock Exploring an unknown graph.
\newblock In: Foundations of Computer Science, 1990. Proceedings., 31st Annual
  Symposium on, IEEE (1990)  355--361

\bibitem{FT08}
Fomin, F.V., Thilikos, D.M.:
\newblock An annotated bibliography on guaranteed graph searching.
\newblock Theor. Comput. Sci. \textbf{399}(3) (2008)  236--245

\bibitem{Albers00}
Albers, S., Henzinger, M.R.:
\newblock Exploring unknown environments.
\newblock SIAM J. Comput. \textbf{29}(4) (2000)  1164--1188

\bibitem{alpern2002theory}
Alpern, S., Gal, S.:
\newblock The theory of search games and rendezvous. Volume~55.
\newblock Kluwer Academic Publishers (2002)

\bibitem{baeza1993searching}
Baeza-Yates, R.A., Culberson, J.C., Rawlins, G.J.E.:
\newblock Searching in the plane.
\newblock Information and Computation \textbf{106} (1993)  234--234

\bibitem{CILP13}
Czyzowicz, J., Ilcinkas, D., Labourel, A., Pelc, A.:
\newblock Worst-case optimal exploration of terrains with obstacles.
\newblock Inf. Comput. \textbf{225} (2013)  16--28

\bibitem{DKP91}
Deng, X., Kameda, T., Papadimitriou, C.H.:
\newblock How to learn an unknown environment (extended abstract).
\newblock In: FOCS. (1991)  298--303

\bibitem{Bellman63}
Bellman, R.:
\newblock An optimal search problem.
\newblock Bull. Am. Math. Soc. (1963)  270

\bibitem{Beck64}
Beck, A.:
\newblock on the linear search problem.
\newblock Israel Journal of Mathematics \textbf{2}(4) (1964)  221--228

\bibitem{demaine2006online}
Demaine, E.D., Fekete, S.P., Gal, S.:
\newblock Online searching with turn cost.
\newblock Theoretical Computer Science \textbf{361}(2) (2006)  342--355

\bibitem{Albers03}
Albers, S.:
\newblock Online algorithms: a survey.
\newblock Math. Program. \textbf{97}(1-2) (2003)  3--26

\bibitem{AS11}
Albers, S., Schmelzer, S.:
\newblock Online algorithms - what is it worth to know the future?
\newblock In: Algorithms Unplugged.
\newblock (2011)  361--366

\bibitem{berman1998line}
Berman, P.:
\newblock On-line searching and navigation.
\newblock In Fiat, A., Woeginger, G., eds.: Online Algorithms The State of the
  Art.
\newblock Springer (1998)  232--241

\bibitem{FKKLT08}
Fleischer, R., Kamphans, T., Klein, R., Langetepe, E., Trippen, G.:
\newblock Competitive online approximation of the optimal search ratio.
\newblock SIAM J. Comput. \textbf{38}(3) (2008)  881--898

\bibitem{DDKPU13}
Dereniowski, D., Disser, Y., Kosowski, A., Pajak, D., Uznanski, P.:
\newblock Fast collaborative graph exploration.
\newblock In: ICALP. Volume to appear. (2013)

\bibitem{CFMS10}
Chalopin, J., Flocchini, P., Mans, B., Santoro, N.:
\newblock Network exploration by silent and oblivious robots.
\newblock In: WG. (2010)  208--219

\bibitem{DFKNS07}
Das, S., Flocchini, P., Kutten, S., Nayak, A., Santoro, N.:
\newblock Map construction of unknown graphs by multiple agents.
\newblock Theor. Comput. Sci. \textbf{385}(1-3) (2007)  34--48

\bibitem{FGKP06}
Fraigniaud, P., Gasieniec, L., Kowalski, D.R., Pelc, A.:
\newblock Collective tree exploration.
\newblock Networks \textbf{48}(3) (2006)  166--177

\bibitem{HKLT13}
Higashikawa, Y., Katoh, N., Langerman, S., ichi Tanigawa, S.:
\newblock Online graph exploration algorithms for cycles and trees by multiple
  searchers.
\newblock J. Comb. Optim. (2012)

\bibitem{WIFBZP11}
Wang, G., Irwin, M.J., Fu, H., Berman, P., Zhang, W., Porta, T.L.:
\newblock Optimizing sensor movement planning for energy efficiency.
\newblock ACM Transactions on Sensor Networks \textbf{7}(4) (2011) ~33

\bibitem{beauquier2010utilizing}
Beauquier, J., Burman, J., Clement, J., Kutten, S.:
\newblock On utilizing speed in networks of mobile agents.
\newblock In: Proceeding of the 29th ACM SIGACT-SIGOPS Symposium on Principles
  of distributed computing, ACM (2010)  305--314

\bibitem{CGKK11}
Czyzowicz, J., Gasieniec, L., Kosowski, A., Kranakis, E.:
\newblock Boundary patrolling by mobile agents with distinct maximal speeds.
\newblock In: ESA. (2011)  701--712

\bibitem{KK12}
Kawamura, A., Kobayashi, Y.:
\newblock Fence patrolling by mobile agents with distinct speeds.
\newblock In: ISAAC. (2012)  598--608

\end{thebibliography}
\bibliographystyle{splncs}


\newpage

\begin{appendix}

\section{Appendix}

\subsection{Proof of Lemma~\ref{lem: all properties}}

\begin{proof} By the observation made in the preliminaries we assume that the segment may be partitioned into subsegments, such that each subsegment is searched by only one robot of the collection.

\textbf{(a)}~ Suppose a robot $r_i$  searches the non contiguous subintervals $[a_1,b_1]$ and $[a_2,b_2]$ (with $b_1 <a_2$), of the unit interval $[0,1]$. We modify the schedule so that robot searches the interval $[a_2 - (b_1-a_1),b_2]$. The time robot $r_i$ stops searching remains the same, as do the finishing times for the rest of the robots, once we shift the allocated searching intervals that fall between $[a_1,b_1]$ and $[a_2,b_2]$.

\textbf{(b)}~ Suppose some robot $r_i$ has an idle period before it searches its last allocated point of the domain. We can eliminate this period by switching the robot to a moving mode (either walking or searching) earlier, which reduces its individual finishing time. Hence, we may assume that all robots have idle times only after the time they finish searching. Now consider a robot $r_i$ that finishes searching its unique (due to part (a)) interval $[a,b]$ strictly earlier than the rest of the robots, by, say, $\epsilon$ time units. We can then reschedule robot $r_i$ so as to search $[a-\epsilon s_i/2, b+\epsilon s_i/2]$. Robots searching a preceding interval now search a subinterval that may have been shortened (but not lengthened), and they do not walk more. Robots that search succeeding intervals may have their searching intervals shortened, in which case they may need to process some subinterval by walking instead of searching. Since for each robot the walking speed is strictly higher than the searching speed, this process can only reduce the total finishing time. The argument is similar if $[a,b]$ above lies in one of the endpoints of the domain $[0,1]$. 

\textbf{(c)}~ This is true, since otherwise a robot would have 0 searching time which would contradict part (b). 

\textbf{(d)}~
By part (a) and (c) above, the domain is partitioned into subintervals of length $c_1,\ldots, c_n$ with the understanding that $c_i$ is searched by robot $r_i$. 

In what follows, we investigate the effect of switching the order of two robots that search two consecutive subintervals, so that the union of the intervals remains unchanged. In particular we will redistribute the portion of the union of the two intervals that each robot will search, enforcing the optimality condition of part (b). Since we will only redistribute the length of intervals $i,j$ to robots $i,j$, the rest of the subintervals will remain the same, and so will the finishing search times of the remaining robots. 

Without loss of generality, assume that interval $c_i$ lies in the leftmost part of the domain from which all robots start (we may assume this since any preceding robots will not be affected as we maintain the union of the intervals that both robots $r_i,r_j$ will search together). Note that robot $r_i$ searches $c_i$ while robot $r_j$ walks $c_i$ and searches $c_j$. By part (b) all robots have the same finishing time, so we have
\begin{equation}
 \frac{c_i}{s_i} = \frac{c_i}{w_j} + \frac{c_j}{s_j} \label{equa: equal fin time for two}
\end{equation}
If $c_i = \lambda (c_i+c_j)$ (in which case $c_j = (1-\lambda)(c_i+c_j)$), then substituting in ~\eqref{equa: equal fin time for two} and solving for $\lambda$ gives 
$$ \lambda = \frac{1}{\left(\frac{1}{s_i}+\frac{1}{s_j}-\frac{1}{w_i}\right) s_j}.$$
Hence, we conclude that the finishing time for both robots is
$$T = \frac{\lambda (c_i+c_j)}{s_i}
	=\frac{c_i+c_j}{\left(\frac{1}{s_i}+\frac{1}{s_j}-\frac{1}{w_i}\right) s_i s_j}.$$
We now reschedule the robots so that robot $r_j$ searches first, say a $\mu$ portion of $c_i+c_j$, and robot $r_i$ searches the remaining (and second in order) subinterval of length $(1-\mu)(c_i+c_j)$. This means that robot $r_i$ will now walk the interval of length $\mu(c_i+c_j)$. Since by part (b) the two robots must finish simultaneously, the same calculations show that the new finishing time is 
$$T' = \frac{\mu (c_i+c_j)}{s_j}
	=\frac{c_i+c_j}{\left(\frac{1}{s_i}+\frac{1}{s_j}-\frac{1}{w_j}\right) s_i s_j}.$$
It is easy to see then that $T' < T$ whenever $w_i > w_j$, concluding what we need.
\end{proof}

\subsection{Online Searching with Robots of Equal Walking Speeds}\label{sec: online with w-uniform}

We call {\em w-uniform} the instance of the \bp~in which all agents have the same walking speeds. Moreover if the searching speeds are the same - the problem is {\em totally uniform}. Clearly all $n$ robots participate in the swarm of the \OnlineSearch~algorithm. Considering the speeds of the schedules obtained by the offline and online algorithms given by Theorem~\ref{thm: opt solution to ssc1} and Theorem~\ref{online-speed}, the upper bound on the competitive ratio $C$ of the \OnlineSearch~algorithm is given by
$$
C= \sup_{R}\frac{T_{\textrm{online}}}{T_{\textrm{opt}}}=
\sup_{R}\left( \sum_{k=1}^{n} s_k \prod_{j=k+1}^{{n}} \left( 1 - \frac{s_j}{w_j} \right)\right)  /
\left( \frac{ \sum_{k=1}^n \frac{1}{\delta_k}}{1+ \sum_{k=1}^n \frac{1}{w_k\delta_k}}\right) 
$$
where the supremum of the ratio is taken over all configurations $R$ of robots' speeds and $\delta_i := \frac{1}{s_i} - \frac{1}{w_i}$, for $i=1,2,\ldots ,k$. 

Observe that this ratio remains the same if the instance of the problem is scaled down to all walking speeds equal to 1. Then the simple calculation shows that the value of the competitive ratio is simplified to

\begin{equation}\label{equa: upper bound comp ration W-uniform}
C= \sup_{R} 
\left( 1-\prod_{k=1}^n(1-s_k) \right)
\left(
1 + 
\frac{1}
{\sum_{k=1}^n\frac{s_k}{1-s_k}}
\right)
\end{equation}

In what follows we compute a numeric upper bound for~\eqref{equa: upper bound comp ration W-uniform}. Such a task seems challenging as it involves $n$ many parameters, i.e. the searching speeds. As the expression is symmetric in the parameters, one should expect that it is maximized when all parameters are the same. Effectively, this would mean that competitive ratio of our algorithm is worst only for totally uniform  instances, i.e. where all searching speeds are the same and all walking speeds are the same. This is what we make formal in the next technical lemma. 

\begin{lemma}\label{lem: ratio maximized for totally uniform}
Given some fixed $n$, expression~\eqref{equa: upper bound comp ration W-uniform} is maximized for totally uniform instances of \bp. 
\end{lemma}

\begin{proof}
Consider the function $f: (0,1)^n \mapsto \reals_+$ defined as 
$$ f(s_1,s_2, \ldots, s_n) = \left( 1-\prod_{k=1}^n(1-s_k) \right)
\left(
1 + 
\frac{1}
{\sum_{k=1}^n\frac{s_k}{1-s_k}}
\right).
$$
A necessary condition for optimality is that $\frac{\vartheta f}{\vartheta s_i} = 0$ for $i=1,\ldots, n$. Towards computing the partial derivatives we introduce the shorthands
$$ Q:= {\sum_{k=1}^n\frac{s_k}{1-s_k}} ~~~\textrm{and}~~~
P := \prod_{k=1}^n(1-s_k),$$
and we observe that 
$$
\frac{\vartheta }{\vartheta s_i} Q = 
\frac{\vartheta }{\vartheta s_i} {\frac{s_i}{1-s_i}} = \frac{1}{(1-s_i)^2}
$$
and that
$$
\frac{\vartheta }{\vartheta s_i} P
=
- \prod_{k=1,\ldots, n ~\&~ k \not = i}^n(1-s_k)
=
{- \frac{1}{1-s_i}} P.
$$
Then, we easily get that 
$$
\frac{\vartheta f}{\vartheta s_i}
=
\frac{1}{(1-s_i)^2}\frac{1}{Q^2}+\frac{1}{1-s_i}P - \frac{-\frac{1}{1-s_i}Q-P\frac{1}{(1-s_i)^2}}{Q^2}.
$$
Requiring that the above partial derivative identifies with 0 and solving for $s_i$ gives
$$ s_i = 1 - \frac{1+P}{Q^2P+Q}.$$
Note that this already shows that all $s_i$ are equal when $f$ is maximized. In order to complete the proof, we need to show that these values of $s_i$ are indeed between 0 and 1. For this we first note that $s_i<1$ since $P,Q>0$, and hence it suffices to show that $s_i$ are positive when $f$ is maximized. 

In order to show that $s_i >0$, we observe that if $Q\geq 1$, then it can be easily seen that $1 - \frac{1+P}{Q^2P+Q} \geq 0$ (independently of the value of $P$). This is because $s_i \geq 0$ if and only if $-1-P+Q+P Q^2 \geq 0$. The later quadratic in $Q$ has 1 as the higher root and therefore is strictly positive for the values of $Q$ that exceed 1. 

It remains to show that $s_i>0$ in the case when $Q<1$. For this we do the following trick. Since $ s_i = 1 - \frac{1+P}{Q^2P+Q}$, we also have $ 1- s_i = \frac{1+P}{Q^2P+Q}$ and so
$$\frac{s_i}{1-s_i} = \frac{1 - \frac{1+P}{Q^2P+Q}}{\frac{1+P}{Q^2P+Q}}.$$
Summing the left-hand-side over $i=1,\ldots,n$ gives exactly $Q$,  so we conclude that for the values of $s_i$ that optimize $f$ we have 
\begin{equation}\label{equa: relation QP when f optimal}
Q = n \frac{Q^2P+Q-1-P}{1+P}.
\end{equation}
Let then $P=\frac{n(1-Q)+Q}{n(Q^2-1)+Q}$, i.e. the value as indicated after we solve for $P$ in \eqref{equa: relation QP when f optimal}. We get then that since $s_i = (-1-P+Q+P Q^2)/(Q^2P+Q)$, its numerator can be written as an expression in $n,Q$ as
$$(1-Q) \left( \frac{n(1-Q)+Q}{n(1-Q^2)+Q}(1+Q) -1\right)$$
One can easily see that $\frac{n(1-Q)+Q}{n(1-Q^2)+Q}$ decreases with $n$, and hence the expression we want to show to be non negative is at least 
$$(1-Q) \left( \lim_{n \rightarrow \infty} \frac{n(1-Q)+Q}{n(1-Q^2)+Q}(1+Q) -1\right) = 0$$
exactly as wanted.
\end{proof}

To conclude, Lemma~\ref{lem: ratio maximized for totally uniform} says that in order to determine the competitive ratio of our algorithm for general w-uniform case, it suffices to consider totally uniform instances. This is what we do in the next subsection.

\subsubsection{Online Searching for Totally Uniform Instances (Proof of Theorem~\ref{thm: summary of online w-uniform})}\label{sec: online for totally uniform}

For the sake of notation ease, we normalize all speeds so as to have uniform walking speeds 1 and uniform searching speeds $\alpha$. 

Following the analysis for w-uniform case, we know that the competitive ratio of our algorithm for the totally uniform instance as described above is 
\begin{equation}\label{equa: def of f(n,a)}
 \frac{T_{\textrm{online}}}{T_{\textrm{opt}}} \leq \frac{\left(\alpha (n-1)+1\right)\left( 1-(1-\alpha)^n \right)}{\alpha n} ~~~~~~~(~:=f_n(\alpha)~ )
\end{equation}
As already indicated, we denote the above expression on $\alpha,n$ by $f_n(\alpha)$. From now on we think of $f_n(\alpha)$ as the competitive ratio of the \OnlineSearch~algorithm. Table~\ref{tab: comp ratio 2,3,4} is easy to establish using elementary calculations and shows the competitive ratio for small number of robots.
\begin{table}[center]
\centering
\begin{tabular}{|l|l|l|}
\hline 
$n$ & $\max_{\alpha} f_n(\alpha)$ & $\argmax_{\alpha} f_n(\alpha)$ \\
\hline 
2 & $\frac{9}{8} = 1.125$ & $\frac{1}{2}=0.5$\\
3 & $\frac{172+7 \sqrt{7}}{162}  \approx 1.17605$ & $\frac{5-\sqrt{7}}{6}  \approx 0.392375$ \\
4 &  $\approx 1.20386$ & $\frac{1}{12} \left(11-\frac{9}{\left(85-4 \sqrt{406}\right)^{1/3}}-\left(85-4 \sqrt{406}\right)^{1/3}\right) \approx 0.322472 $ \\
\hline
\end{tabular}
\caption{Competitive ratio of the \OnlineSearch~Algorithm for the collections of robots of size 2,3,4.}
\label{tab: comp ratio 2,3,4}
\end{table}
In what follows we give a detailed analysis of the competitive ratio. For this we need the next technical lemma. 
\begin{lemma}\label{lem: comp ratio increases with robots}
Let $\alpha_n = \argmax f_n(\alpha)$. Then $f_n(\alpha_n)$ increases with $n$.
\end{lemma}

\begin{proof}
Note that Table~\ref{tab: comp ratio 2,3,4} already shows the lemma for $n\leq 4$. Hence, below we may assume that $n\geq 5$. First we show that $f_n(a)$ has a unique maximizer when $\alpha \in (0,1)$. For this we examine the critical points of $f_n(a)$ by solving $\frac{d~f_n(\alpha) }{d~\alpha} = 0$, i.e. 
\begin{equation}\label{equa: root wrt to a}
r_n(\alpha):= (1-\alpha)^{n-1} \left( 1+\alpha(n-1) + \alpha^2n(n-1) \right) - 1 =0.
\end{equation}
To show that $r_n(\alpha)$ has a unique solution in $(0,1)$ we again take the derivative with respect to $\alpha$ to find that $\frac{d~r_n(\alpha) }{d~\alpha} = (1-\alpha)^{n-2}\alpha n (n-1) (1-\alpha(n+1))$. This means that $r_n(\alpha)$ increases when $\alpha <1/(n+1)$ and decreases otherwise. Noting also that $r_n(0)=0$ and $r_n(1)=-1$, we conclude that $r_n(\alpha)$ has a unique root in $(0,1)$, i.e. $f_n(\alpha)$ has a unique maximizer $\alpha_n$ over $\alpha \in (0,1)$, and in particular $\alpha_n > \frac{1}{n+1}$. Next we will provide a slightly better bound on the roots of $r_n(\alpha)$. For this we observe that 
$$ r_n\left(\frac{1}{n-1} \right)  = \frac{\left(\frac{n-2}{n-1}\right)^n (3 n-2)}{n-2} - 1$$
which can be seen to be positive for $n\geq 5$. 
Hence, by the monotonicity we have already shown for $r_n(\alpha)$, we may assume that its root $\alpha_n$ satisfies
\begin{equation}\label{equa: bound on root}
\frac{1}{n-1}< \alpha_n < 1.
\end{equation}
Now we turn our attention to $f_n(a_n)$. Since $\alpha_n$ satisfies $r_n(\alpha_n)=0$, it is easy to see that
$$ f_n(a_n) =  \frac{(1+\alpha_n (n-1))^2}{1+\alpha_n (n-1)+\alpha_n^2n(n-1)} $$
simply by substituting $(1-\alpha)^{n-1}$ from \eqref{equa: root wrt to a} into~\eqref{equa: def of f(n,a)}. Recalling that $a_n$ is also a function on $n$ we get that
$$ \frac{d~f_n(\alpha_n) }{d~n}
=
\frac{\alpha_n(1-\alpha_n)}{\left( 1+\alpha_n(n-1)+\alpha_n^2 n (n-1)\right)^2}\left(1-\alpha_n^2(n-1)^2\right)
\frac{d~\alpha_n }{d~n}
$$
Due to~\eqref{equa: bound on root} we conclude that $\frac{d~f_n(\alpha_n) }{d~n}$ has the opposite sign of  $\frac{d~\alpha_n }{d~n}$, i.e. for $n\geq 5$ we have that $f_n(\alpha_n)$ increases with $n$ if and only if $\alpha_n$ decreases with $n$. So it remains to show that the roots $\alpha_n$ of $r_n(\alpha)=0$ decrease with $n$. Observe here that at this point we may restrict consideration to integral values of $n$. 

To conclude the lemma we argue that $\alpha_{n+1}<\alpha_{n}$. For this we observe that 
$$
\frac{r_{n+1}(\alpha)+1}{r_{n}(\alpha)+1} 
=
\frac{(1-\alpha)(1+\alpha n+\alpha^2n(n+1))}{1+\alpha(n-1)+\alpha^2n(n-1)}
$$
which is clearly less than 1 for $\alpha>\frac{1}{n+1}$ (just by solving for $\alpha$). Effectively this means that the graph of $r_{n}(\alpha)+1$ is above the graph of $r_{n+1}(\alpha)+1$ for every $\alpha > \frac{1}{n+1}$, and hence 
$$ r_n(\alpha_{n+1}) > r_{n+1}(\alpha_{n+1}) = 0 = r_n(\alpha_n) > r_{n+1}(\alpha_n).$$
But then, from the monotonicity we have shown for $r_n(\alpha)$ this implies that $\alpha_{n+1}<\alpha_n$ as wanted. 
\end{proof}

The next lemma in combination with Lemma~\ref{lem: ratio maximized for totally uniform} prove Theorem~\ref{thm: summary of online w-uniform}. 

\begin{lemma}\label{lem: comp ratio two extremes}
For the totally uniform instances of the \bp, the \OnlineSearch~algorithm has competitive ratio at most 9/8 for two robots, and the competitive ratio of at most $\max_c \{ (1+1/c)(1-e^{-c})\} \approx 1.29843$ for any number of robots. 
\end{lemma}

\begin{proof}
By the proof of Lemma~\ref{lem: comp ratio increases with robots}, $a_2$ satisfies $(1-\alpha_2)(1+\alpha_2 +2\alpha_2^2)=1$, which has the unique solution $\alpha_2 = 1/2$. It is easy to see then $f_2(1/2)=9/8$. 

Next, by Lemma~\ref{lem: comp ratio increases with robots}, the bigger is the number of robots, the higher is the competitive ratio of our algorithm. Hence, we need to determine $\lim_{n \rightarrow \infty} f_n(a_n)$. 

To that end we note that if $a_n = o(1/n)$, then $1-(1-\alpha)^n \approx \alpha n$, and so 
$$ \lim_{n \rightarrow \infty} 
 \frac{\left(\alpha (n-1)+1\right)\left( 1-(1-\alpha)^n \right)}{\alpha n}
= \lim_{n \rightarrow \infty} \frac{\alpha n - \alpha + 1}{\alpha n} \alpha n = 1.
$$
Similarly, if $a_n = \omega(1/n)$, then $1-(1-\alpha)^n$ tends to 1 as $n$ goes to infinity. Consequently, 
$$ \lim_{n \rightarrow \infty} 
 \frac{\left(\alpha (n-1)+1\right)\left( 1-(1-\alpha)^n \right)}{\alpha n}
= \lim_{n \rightarrow \infty} \frac{\alpha n - \alpha + 1}{\alpha n} = 1.
$$
It remains to check what happens when $\alpha = c/n$ for some $c \in \reals_+$. But then 
$$ \lim_{n \rightarrow \infty} 
 \frac{\left(\alpha (n-1)+1\right)\left( 1-(1-\alpha)^n \right)}{\alpha n}
=  (1+1/c)(1-e^{-c}).
$$
The last expression is maximized when $c \approx 1.79328$ and the value it attains approaches 1.29843.
\end{proof}

A picture for the rate of growth of the competitive ratio of the crawling algorithm is depicted in Figure~\ref{fig: CRplot}.
\begin{figure}[htb]
\label{fig: CRplot}
\centering
\includegraphics[height=5cm]{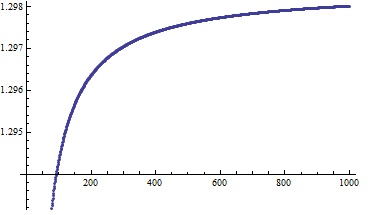}
\caption{Competitive ratio of the \OnlineSearch Algorithm as a function of the number of robots}
\end{figure}

\end{appendix}

\end{document}